\newif\ifarxiv

\arxivtrue

\ifarxiv
\documentclass{article}
\usepackage{arxiv}
\usepackage{amsthm}
\else
\documentclass[10pt,twocolumn,twoside]{IEEEtran}

\def\QEDclosed{\mbox{\rule[0pt]{1.3ex}{1.3ex}}} 

\def\QED{\QEDclosed} 

\fi

\usepackage{generic}
\usepackage{cite}
\usepackage{amsmath,amssymb,amsfonts,todonotes}
\usepackage{algorithmic}
\usepackage{graphicx}
\usepackage{textcomp}
\usepackage{url}
\usepackage{wrapfig}

\newtheorem{theorem}{Theorem}
\newtheorem{lemma}{Lemma}

\newtheorem{remark}{Remark}
\newtheorem{definition}{Definition}

\newtheorem{assumption}{Assumption}

\usepackage{bbm}

\newcommand{\R}{\mathbb{R}}
\newcommand{\N}{\mathcal{N}}
\newcommand{\A}{\mathcal{A}}
\newcommand{\X}{\mathcal{X}}

\newcommand{\C}{\mathcal{C}}
\newcommand{\n}{\{n\}}
\newcommand{\eig}{\text{eig}}
\newcommand{\hC}{\widehat{\mathcal{C}}}

\newcommand{\Tb}{\mathbf{T}}
\newcommand{\Scal}{\mathcal{S}}

\newcommand{\abold}{\bm{\alpha}}
\newcommand{\Vhat}{\widehat{V}}
\newcommand{\Phat}{\widehat{P}}

\DeclareMathOperator*{\argmax}{arg\,max}
\usepackage[ruled,vlined,linesnumbered ]{algorithm2e}
\usepackage{bm}

\begin{document}
\title{Clustered Control of Transition-Independent MDPs}

\author{Carmel Fiscko$^{1}$, Soummya Kar$^{1}$, and Bruno Sinopoli$^{2}$
\thanks{$^{1}$Carmel Fiscko and Soummya Kar are with the Dept. of Electrical and Computer Engineering at Carnegie Mellon University at 5000 Forbes Ave, Pittsburgh, PA 15213. {\tt\small cfiscko@andrew.cmu.edu, soummyak@andrew.cmu.edu}. This material is based upon work supported by the National Science Foundation Graduate Research Fellowship Program under Grant No. DGE1745016. Any opinions, findings, and conclusions or recommendations expressed in this material are those of the author(s) and do not necessarily reflect the views of the National Science Foundation. Additional support from the Hsu Chang Memorial Fellowship in ECE.}%
\thanks{$^{2}$Bruno Sinopoli is with the Dept. of Electrical and Systems Engineering at Washington University in St. Louis, MO at 1 Brookings Dr, St. Louis, MO 63130. {\tt\small bsinopoli@wustl.edu }}%
}

\maketitle

\begin{abstract}
This work studies efficient solution methods for cluster-based control policies of transition-independent Markov decision processes (TI-MDPs). We focus on control of multi-agent systems, whereby a central planner (CP) influences agents to select desirable group behavior. The agents are partitioned into disjoint clusters whereby agents in the same cluster receive the same controls but agents in different clusters may receive different controls. Under mild assumptions, this process can be modeled as a TI-MDP where each factor describes the behavior of one cluster. The action space of the TI-MDP becomes exponential with respect to the number of clusters. To efficiently find a policy in this rapidly scaling space, we propose a clustered Bellman operator that optimizes over the action space for one cluster at any evaluation. We present Clustered Value Iteration (CVI), which uses this operator to iteratively perform ``round robin" optimization across the clusters. CVI converges exponentially faster than standard value iteration (VI), and can find policies that closely approximate the MDP’s true optimal value. A special class of TI-MDPs with separable reward functions are investigated, and it is shown that CVI will find optimal policies on this class of problems. Finally, the optimal clustering assignment problem is explored. The value functions TI-MDPs with submodular reward functions are shown to be submodular functions, so submodular set optimization may be used to find a near optimal clustering assignment. We propose an iterative greedy cluster splitting algorithm, which yields monotonic submodular improvement in value at each iteration. Finally, simulations offer empirical assessment of the proposed methods.

\end{abstract}

\section{Introduction}
\label{sec:introduction}


Multi-agent systems are ubiquitous in natural and engineered systems, and understanding how the decision processes of the individuals translate to group behavior has been a key goal of research study. Influencers aiming to achieve global control objectives on the system must also characterize how their controls affect agent decision policies, and determine which controls induce the desired long-term group behavior.

Social media, robot swarms, and markets are examples of systems whose dynamics are driven by the decisions of individuals. A person's posts online reflect their own opinions, yet they are influenced by their friends, trends, and current events. These decisions, in turn, affect the path of future posts and trends. Within cooperative robotics, a central objective must be completed, yet each robot's path must consider those of its peers. On the competitive side, businesses weigh their rivals' actions to likewise position themselves to an advantage. Any regulatory power that wishes to achieve control objectives on the system must therefore understand the agents' behavior.


\subsection{Multi-Agent Systems}
Such groups of agents are often modeled as multi-agent systems (MAS), where groups of intelligent entities make decisions while interacting within some environment \cite{stone2000multiagent}. These agents can model applications such as computer networks, smart grids, and cyberphysical systems, in addition to the aforementioned models of social media, robots, and finance \cite{dorri2018multi}. Natural questions arising from these examples include control, security, coordination, and scalability. 

In engineered models of MAS, a common question is how should each agent act optimally to achieve a personal objective within the interconnected system. In a game theoretic model, an individual is endowed with a utility function that rates the goodness of a potential action in the current world \cite{osborne1994course}. Outcomes of games are generally characterized by equilibria, in which agents self-reinforce game-play at the same action, which may or may not have a notion of optimality. 

Multi-agent reinforcement learning (MARL) was born from the field of Markov or stochastic games \cite{shapley1953stochastic}, extending the idea of stochastic repeated play games for agents to learn within a probabilistic environment, as presented in the seminal work \cite{littman1994markov}. Many contributions have investigated how agents may learn policies in this setting under a variety of assumptions, structures, and guarantees of optimality. The objective of the agents can generally be described as cooperative \cite{boutilier1996planning}, competitive \cite{littman1994markov}, or general-sum \cite{littman2001value}. Various agent communication structures have been proposed, ranging from fully centralized \cite{chen2021communication} to decentralized \cite{kar2013cal}. Techniques for agents to learn policies generally fall under value-based methods where agents estimate q-functions \cite{hu2003nash}, or policy-based methods such as fictitious play \cite{brown1951iterative}.

This objective of global control of a multitude of heterogeneous networked agents has been widely explored in domain-specific applications. MARL \cite{yang2019two} and deep reinforcement learning (RL) \cite{gao2021consensus} have been used to address the management of power systems distribution networks. In economic applications, MARL has been used to analyze market makers providing liquidity \cite{ganesh2019reinforcement}, and creating tax policies that balance equality and productivity \cite{zheng2020ai}. A few social objectives that have been considered include the issues of autonomous vehicle adoption \cite{bonnefon2016social} and responses to climate change \cite{janssen1998battle}. MAS theory has also been used to propose traffic signal control \cite{arel2010reinforcement} and scheduling of taxis \cite{glaschenko2009multi}. 

In this work, we consider a MAS with agents who have learned policies. These policies may be modeled formally via game theory or MARL, or they may simply be learned via frequentist estimation of observed agent behavior. Depending on the guarantees of the found policies, if any exist, the long-term behavior of the MAS may result in an equilibrium at one state, cyclical behavior between a subset of states, or some stationary distribution over the entire state space.

\subsection{Transition-Independent MDPs}
If the agents sample their next states and actions independently given the current state and action, then the transition kernel of the MDP may be expressed as the product of the individual agent transitions. MDPs that satisfy this structure are called TI-MDPs \cite{becker2003transition}. If in addition the reward function is additive with respect to agents, the model is called a factored MDP \cite{osband2014near}. These models have been widely studied for their structural properties and complexity reduction over standard MDPs. Solutions for this formulation include approximate policy iteration \cite{koller2013policy}, graphical methods \cite{sallans2004reinforcement}, as well as approximate transition independence \cite{sahabandu2021scalable}. Efforts to reduce the scope of policy search include efficient policy sampling methods \cite{osband2014near}, relaxations to deterministic policies \cite{ng2013pegasus}, and hierarchical methods \cite{dietterich2000hierarchical}. 

\subsection{Controlled MAS}

\textbf{The objective in this work is to characterize the ability of a central planner (CP) to change this long-term behavior exhibited by the agents.} The CP takes advantage of the learned agent policies in that these policies are dependent on the actions of the agents, any environmental states, and the CP signal; the CP can thus select controls to increase the probability that the agents sample ``good" future actions. There are two levels of decision-making: each agent chooses actions to accomplish the agent's goal, and the CP influences the agents by making certain actions seem ``better" and thus more likely to be chosen. The goal of the CP may be aligned with, or completely different from, the goals of the agents. This model aims to characterize mathematically the role of CPs in MAS, entities that can affect the decision of a large number of agents. These CPs may be policy makers, e.g. large social networks like Twitter and Facebook, or single individuals who are followed by many users.

Under assumptions of Markovian dynamics, we model the controlled MAS as a Markov decision process (MDP) \cite{Fiscko2019ControlOP}. The difference between this formulation and MARL is that the goal of our CP is to control system behavior given the agents' \emph{a priori} learned, i.e. time-homogeneous, policies; MARL solves the question of how agents should simultaneously learn their policies. This formulation is similar to a MARL model where $N-1$ agents have fixed but possibly non-optimal policies and the $N^{\text{th}}$ agent learns a policy under global visibility properties. This formulation enables flexibility as the CP is agnostic to the specific process by which each agent learned their policy: agents may be competitive/cooperative/general-sum, acting optimally or sub-optimally, highly/sparsely connected, using policy- or value-based methods, or a heterogeneous mix of all of the above. This flexibility motivates defining a MDP from the CP's perspective.

\subsection{Cluster-Based Control Policies}
We explored the controlled MAS problem first in the scenario where all agents received the same broadcast signal from the CP \cite{Fiscko2019ControlOP}. Next, the problem was generalized to the case where the CP could broadcast unique signals to each agent \cite{fiscko2021efficient}. In that work, we note that finding agent-personalized policies becomes computationally infeasible as the problem size scales; our proposed solution is to partition the agents into sets and have the CP design control policies per-cluster rather than per-agent.



Clustering techniques are of importance in many engineering fields for their scale and control properties, such as in wireless sensor networks \cite{abbasi2007survey} and swarm robotics \cite{dos2012distributed}.


Here we emphasize that we take a node-based definition of the clusters rather than an edge-based grouping such as \cite{malliaros2013clustering}; agents in the same cluster receive the same signals from the CP, but agents in different clusters may receive different signals. The clustering motivation is twofold: first is that natural cluster definitions may appear in the application, and a practitioner may want to personalize controls for each type in accordance to system goals and constraints. Secondly, more degrees of freedom in control can yield better performance in terms of value and computation speed, even amongst agents for whom no clear partitions exist. 


If no cluster assignment is given \emph{a priori}, we desire an optimal clustering that maximizes the CP's value function subject to a desired number of clusters. While the general form of this optimization problem is known to be NP hard, approximations have been studied across many applications. Variants have been studied \cite{zhao2005greedy}, but these are concerned with minimization with non-negativity constraints preventing analysis of maximization problems. Equivalently, the problem may be stated as an objective subject to a matroid constraint, for which there exists a $(1-1/e)$ approximation  \cite{calinescu2007maximizing}.

\subsection{Our Contributions}
In this work study the \emph{cluster-based TI-MDP control problem}, both in terms of efficient policy computation and the clustering assignment. We introduce Clustered Value Iteration (CVI), an approximate ``round robin" method for solving for cluster-based MDP policies whose per-iteration complexity is independent of the number of clusters, whereas standard VI has exponential dependence. We show that the value function of TI-MDPs with separable reward functions yields a separable structure that enables CVI to find optimal policies for this class of problems. Next, we investigate an approximate method for the optimal cluster assignment problem of the agents. We show that TI-MDPs with submodular reward functions have submodular value functions, allowing agent assignment to be solved as a submodular set optimization problem. We propose an iterative greedy splitting method which gives a structured search to the clustering problem with provable improvement in value. The optimal values of these suboptimal clustering assignments are shown to form a submodular lower bound for the optimal values of the optimal clustering assignments.

We outline the problem formulation in Section \ref{probelm statement}. Section \ref{CVI section} describes CVI and its convergence and complexity properties. Section \ref{separable section} explores TI-MDPs with separable reward functions. The clustering assignment problem and ideas of submodularity are discussed in Section \ref{cluster section}. Finally, Section \ref{sims} demonstrates the proposed algorithms on a variety of simulations. Due to page limitations, some proofs are omitted and can be found in the extended version \cite{fiscko2022cluster}.

\subsection{Comparison to Other Behavior Control Paradigms}

Clustered policies are thematically related to state aggregation methods \cite{bertsekas1995dynamic}, but the technical formulations are distinct. In state aggregation, similar states are grouped together to form a new system with a reduced state space and weighted transitions. In comparison, the clustered formulation works for applications with partitioned substates and factored transitions.

The objective of changing the behavior of a MAS has been studied in other domains such as economic applications, where a business may desire an optimal outcome, or a regulator may desire maximal welfare. This is studied in mechanism design \cite{groves1973incentives},  \cite{roughgarden2010algorithmic}, where an economist may design agent interaction properties like information and available agent actions to induce a desired outcome. In comparison, this paper considers a controller acting on an existing system, thus likewise inheriting the existing infrastructure and constraints.

Incentive design is also related \cite{bacsar1998dynamic}, where rewards are offered to agents to encourage them to select different actions. Incentive design can be connected with Stackelberg games \cite{sanjari2022incentive}, whereby a leader offers rewards and the agents react accordingly. Intervention design is another related paradigm that controls networked agents to induce a desired effect. For example, \cite{ballester2006s} considers removing an agent from a system to optimally improve aggregate network activity. Similar network-based interventions aim to maximize system welfare via designed incentives \cite{galeotti2020targeting}, and subject to random graphs \cite{parise2019graphon}. In comparison, the work in this paper does not model the utility functions of the agents, and only assumes the agent policies are Markovian and conditionally independent given the state and CP signal. These mild assumptions mean that the MDP formulation can model systems whose dynamics evolve as the result of a game, but the control results are not based on properties like utility structure. In addition, the techniques we present are not constrained to a specific structure of network influence, i.e. aggregation, and we do not presuppose a method for CP influence beyond changing probabilities, i.e. we do not assume a specific incentive model.

\section{Problem Statement} \label{probelm statement}
\subsection{Multi-Agent MDPs}
Consider a multi-agent system modeled as a \emph{Markov Decision Process} (MDP) $\mathcal{M}=(\X, \A, R, T, \gamma)$ consisting of the state space of the system $\X$, the action space of the CP $\A$, a reward function of the CP $R:\X\times\A\to\mathbb{R}$, a probabilistic state-action to state transition function $T:\X\times\A \to\Delta(\X)$, and a discount parameter $\gamma\in(0,1)$.

The MDP must next be related to the MAS. Consider a finite set of agents $\N = \{1,\dots, N\}$. At any instance, each agent $n\in\N$ may be associated with some environmental state $s_n\in S_n$ and/or a personal action $a_n\in A_n$. The overall behavior of each agent is denoted by $x_n=\{s_n, a_n\}$ with space $X_n=S_n\times A_n$. A state of the MDP is thus the behavior across all the agents: $x=\{x_1,\dots, x_N\}$, and the state space is $\X = \bigotimes_{n\in\N}X_n$. In general, a state written with a set subscript such as $x_{\mathcal{B}}$ refers to the actions realized in state $x$ by the agents in $\mathcal{B}$, i.e. $x_{\mathcal{B}} = \{x_b|b\in\mathcal{B}\}$, and the notation $- \tilde{n}$ will refer to the set $\{{n}| {n}\in \N, {n}\neq \tilde{n}\}$. 




In this work we consider a \emph{clustered formulation}, in which the agents are partitioned into $C$ disjoint clusters and controls are transmitted to each cluster. A \emph{clustering assignment} of the agents is the set of sets $\C = \{\C_1,\dots, \C_C\}$ such that $\C_u\cap \C_v = \emptyset\ \forall u,v \in [1,C], u\neq v$ and $\cup_{\C_c\in \C} = \N$. The cluster assignment of agent $n$ will be written as $\C(n)$.

In this setup, the CP can broadcast a unique signal $\alpha_c\in\A_c$ to each cluster where $\A_c$ is a finite set of options. The overall control space is thus $\A = \bigotimes_{c\in\C} \A_c$ where one control is $\abold = \{\alpha_1,\dots,\alpha_C\}$. Bold $\abold$ will always refer to the vector of controls across the clusters, and $\alpha$ will refer to an element within the vector. For ease of notation, the CP action seen by agent $n$ will be denoted by $\alpha_n$ instead of $\alpha_{\C(n)}$. %

The next element of the MDP is the \emph{state transition function}, which defines the state-action to state densities in the form $p(x'|x,\abold)$. By design, the agents in some cluster $c$ only see the action assigned to their cluster, $\alpha_c$. Given the agent decision processes, further structure on the transitions may be inferred.

\begin{definition}
An \emph{agent policy} $\omega_n: \X\times \A_{\C(n)}\to\Delta(A_n)$ describes the decision-making process of an agent by defining a distribution over the agent's next actions given the current system state and CP's signal to the agent.
\end{definition}

\begin{definition}
If there exist environmental states, such as in the MARL framework, then the \emph{agent environmental transition} $\psi_n: \X\times\A_{\C(n)}\to\Delta(S_n)$ defines a distribution over the agent's next environmental state given the current system state and the CP's signal to the agent.
\end{definition}

\begin{definition}
The \emph{agent behavior} $\phi_n: \X\times\A_{\C(n)}\to\Delta(\X_n)$ defines the system state-action to agent state transitions. If there are no environmental states, then the agent behavior is equivalent to the agent policy. If there are environmental states, then the agent behavior is the product of the agent's policy and its environmental transition:
\ifarxiv
\begin{equation}
\begin{gathered}
    p(x_n^{t+1}|x^t,\alpha_n^t)=p(s_n^{t+1}, a_n^{t+1}|s^t,a^t,\alpha_n^t)=p(a_n^{t+1}|s_n^{t+1},s_n^t,a_n^t,\alpha_n^t)p(s_n^{t+1}|s^t,a^t,\alpha_n^t)\\
    =p(a_n^{t+1}|s_n^t,a_n^t, \alpha_n^t)p(s_n^{t+1}|s^t,a^t,\alpha_n^t)=\omega_n^t\psi_n^t
\end{gathered}
\end{equation}
\fi
\end{definition}

\begin{assumption} \label{as: agent behavior}
The agents' policies, environmental transitions, and behaviors are Markovian and time-homogeneous:

\ifarxiv
\begin{equation}
    p(x_n^{t+1}|x^t,\dots, x^0, \alpha_n^t,\dots,\alpha_n^0)=p(x_n^{t+1}|x^t,\alpha_n^t),\ \forall\  \alpha_n\in \A_{\C(n)},\ x_n\in X_n,\ n\in \N,\ x\in\X,\  t\geq 0.
\end{equation}
\fi
Furthermore, each agent's policies, environmental transitions, and behaviors are independent of the CP actions assigned to the other agents:

\ifarxiv
\begin{equation}
    p(x_n'|x,\abold)=p(x_n'|x,\alpha_n),\ \forall\ \alpha_n\in \A_{\C(n)},\ x_n\in X_n,\ n\in \N,\ x\in\X,\  t\geq 0. 
\end{equation}
\fi
\end{assumption}

These decision making processes are assumed to follow the standard Markov property. Furthermore, the time homogeneity property means that the agents have learned their decision processes \emph{a priori}, be it from MARL, game theory, or another paradigm. The CP is agnostic to the learning processes used by the agents as long as they satisfy the Markov and time-homogeneity assumptions.

These are the key assumptions enabling the CP to take advantage of the agents' \emph{learned} behaviors. While the set of agent behavioral distributions are fixed over time, the CP changes $\abold$ to change the distributions from which the agents sample, thus also changing the sequence of realized $x$. In this way the agents do not learn \emph{new} behavior, but react to the CP with their \emph{existing} behavior.

MDPs that satisfy Assumption \ref{as: agent behavior} have the following structure in their transition kernels.

\begin{definition}\label{TI}
A \emph{transition independent MDP} (TI-MDP) is a MDP whose state transition probabilities may be expressed as the product of independent transition probabilities,
\begin{align}
    p(x'|x,\abold) = \prod_{i\in\mathcal{I}} p(x_i'|x,\alpha_i).\label{factored structure}
\end{align}
The TI-MDP is ``agent-independent" if  $\mathcal{I}$ is equivalent to the set of agents, and ``cluster-independent" if  $\mathcal{I}$ is equivalent to the set of clusters. Note that factorization across the agents implies factorization across the clusters.
\end{definition}

TI-MDPs are closely related to \emph{factored MDPs}, which also have a structure such as \eqref{factored structure} assumed on their transition matrix \cite{guestrin2003efficient}, \cite{osband2014near}. Factored MDPs are usually derived as the result of a dynamic Bayesian network (DBN) such that: $p(x'|s,\abold) = \prod_{n\in \N}p(x_n'|x_{u(n)}, \alpha_n)$ where $u(n)$ is the set of parents of agent $n$ in the graph. In addition, factored MDPs defines a separable reward as the sum of local reward functions scoped by the same graph $r(s,\abold) = \sum_{n\in\N}r_n(x_{u(n)},\alpha_n)$. 

In the rest of Section \ref{probelm statement} and Section \ref{CVI section} we use the TI-MDP structure as in Definition \ref{TI} with no assumptions on the structure of the reward function. Section \ref{separable section} will investigate the assumption of separability of the reward function.

Definition \ref{TI} describes MASs where each agent makes their decision independently after observing the current state and control. This includes general non-cooperative games \cite{osborne1994course}. For example, in best response style game-play, each player $n$ chooses an action that maximizes their utility function $u$ based on the last round of game-play, i.e. $a_n' \in \argmax u(a,\alpha_n, a_n')$. Knowledge of the players' utilities allows the controller to model each agent's decision process  as a probability function $p(a_n'|a,\alpha_n)$. Therefore, any overall action transition $a\to a'$ can be described as the product of $N$ independent factors that are each conditioned on $a$ and $\abold$. In comparison, communication between agents after observing $(a,\abold)$ and before setting a probability distribution over $a_n'$ may introduce dependence between agents; however, disjoint communication sets may still enable transition independence between clusters.

Finally, the reward function of the MDP must be defined.

\begin{assumption}\label{as:reward}
The reward function $r(x,\abold)$ is non-negative, deterministic, and bounded for all $x\in\X$, $\abold\in\A$.
\end{assumption}

Assumption \ref{as:reward} means that the CP has constructed their goal \emph{a priori}, and can encode this known objective into the reward function. For example, an indicator reward can place a reward of one on desirable states.

Finally, the CP may solve the MDP for some policy  $\Pi:\X\to\A$. In this work we consider the standard \emph{value function}:

\small
\ifarxiv
\begin{equation}
    V^{\Pi}(x) = \mathbb{E}\left[\sum_{k=0}^{\infty}\gamma^k r(x_k,\abold_k)\vert x_0=x, \abold_k\sim \Pi(x_k), x_{k+1}|x_k,\abold_k\sim T \right].
\end{equation}
\fi

\normalsize

For brevity, the notation $V\in\R^{|\X|}$ will refer to the vector of values $V(x)\ \forall\ x\in\X$. The policy notation refers to the tuple $\Pi(x)=\{\pi_1(x),\dots,\pi_C(x)\}$ where $\Pi(x) \in\A$ and $\pi_c(x)\in \A_c$. An optimal policy $\Pi^*$ is one that maximizes the value function, $\Pi^*(x) \in \argmax_{\Pi} V^{\Pi}(x)$. The optimal value $V^*(x)$ is known to be unique, and for finite stationary MPDs there exists an optimal stationary deterministic policy \cite{bertsekas1995dynamic}.
\begin{definition}
The \emph{Bellman operator} $\Tb$ applied to the value function $V(x)$ is,
\begin{align}
    \Tb V(x) &=  \max_{\abold\in \A} \mathbb{E}\left[ r(x,\abold)+\gamma V(x') \Big\vert x'\vert x,\abold\sim T \right]\label{bellman}.
\end{align}
\end{definition}

The optimal value function satisfies $\Tb V^*(x) = V^*(x)$.

An immediate result of the clustered MDP setup is that a better optimal value is attainable if $C>1$, as this increases the degrees of freedom for the control. This idea is formalized in the following lemma.

\begin{lemma}\label{le:improvement}
Consider a MDP whose agents have been partitioned into $C$ clusters of arbitrary agent assignment. The optimal values achieved from $C=1$ and $C=N$ are the lower and upper bounds on the optimal values achieved by arbitrary $C$. 

\begin{equation}
    V^*_{C=1}(x) \leq V^*_{1<C<N}(x) \leq V^*_{C=N}(x).
\end{equation}
\end{lemma}

\begin{proof}
See \cite{fiscko2021efficient}.
\end{proof}

This initial result shows that any arbitrary clustering assignment can improve the optimal value of the MDP. This finding motivates further investigation into clustered policies.

\section{Solving for a Clustered Control Policy} \label{CVI section}
The objective is to attain the improved value with a clustered-based control, but this can only be done if an acceptable policy is found. Standard solution techniques like VI and PI rely on the Bellman operator \eqref{bellman}, which maximizes over the entire action space. The cluster-based policy formulation dramatically increases the size of the action space, meaning that exhaustive search methods become intractable. For example, let each $|\A_c|=M$; then $\abold$ has  $|\A|=M^C$ possible options. PI has a complexity of $O(|\A| |\X|^2+|\X|^3)$ per iteration, and VI has a complexity of $O(|\A| |\X|^2)$ per iteration \cite{littman2013complexity}.  The computation time for either method grows exponentially with the number of clusters. This problem is exacerbated in regimes when the model is unavailable, as methods like Q-learning and SARSA only guarantee optimality with infinite observations of all state-action pairs. The conclusion is that an alternate method for solving for a cluster-based $\Pi^*$ is needed.

In this work we take advantage of the structural properties of TI-MDPs and present a clustered Bellman operator that only optimizes the action for one cluster at each call. This operator is used in an algorithm called Clustered Value Iteration (CVI) as shown in Algorithm \ref{CVI}. This algorithm differs from standard VI in that it takes a ``round-robin" approach to optimize the actions across all the clusters. While one the action for one cluster is being optimized, the actions across the other clusters are held constant until their turn to be improved at future iterations. The intuition is that the controller can improve cluster-by-cluster with greatly reduced search time.

\subsection{Algorithm Statement}

There exists an optimal stationary deterministic policy for this MDP, so our focus is constrained to deterministic policies $\Pi$. Let $\pi_c^k(x) = \alpha_c^k$ denote the action to cluster $c$ at iteration $k$ for state $x$. The vector across all clusters is $\Pi^k(x) = \abold^k$. With some new notation, the Bellman operator can be defined with respect to one element of $\abold$.

\begin{definition}
Let $\Pi_{- c}(x) = \abold_{- c}$ be the tuple of $\alpha$ without the element specified by the subscript.
\begin{align}\small
    \Pi_{- c}(x) &\triangleq \{\pi_1(x),\dots, \pi_{c-1}(x),\pi_{c+1}(x),\dots, \pi_C(x)\},\\
    \abold_{- c} &\triangleq \{\alpha_1,\dots,\alpha_{c-1},\alpha_{c+1},\dots,\alpha_C\}.
\end{align}
\end{definition}

\begin{definition}
For some $\Pi(x)$, $V\in\R^{|\X|}$, and $x\in\mathcal{X}$, define $\Tb^{c}_{\Pi}:\R^{|\X|}\to\R$ as the \emph{clustered Bellman operator} that optimizes $\alpha_c\in\A_c$ for fixed $\Pi_{- c}(x)$.
\ifarxiv
\small\begin{align} 
    \Tb^{c}_{\Pi}V(x) &\triangleq \max_{\alpha_c\in\A_c}\mathbb{E}\left[r(x,\{\Pi_{- c}(x),\alpha_c\})+\gamma V(x')\Big\vert x'\vert x,\{\Pi_{- c}(x),\alpha_c\}\sim T\right],\nonumber\\
    &=\max_{\alpha_c\in\A_c}\sum_{x'\in\X} p(x'|x,\{\Pi_{- c}(x),\alpha_c\})(r(x,\{\Pi_{- c}(x),\alpha_c\})+\gamma V(x'))\label{cvibell}.
\end{align}
\fi
\end{definition}
\begin{definition}
For some $\Pi(x)$, $V\in\R^{|\X|}$,  and $x\in\mathcal{X}$, define $\Tb_{\Pi}:\R^{|\X|}\to\R$ as the \emph{clustered Bellman evaluation},

\small\begin{align}\label{eval}
    \Tb_{\Pi}V(x) &\triangleq \mathbb{E}\left[r(x,\Pi(x))+\gamma V(x')\Big\vert x'\vert x,\Pi(x)\sim T\right],\nonumber\\
    &= \sum_{x'\in\X}p(x'|x,\Pi(x))(r(x,\Pi(x))+\gamma V(x')).
\end{align}
\end{definition}

\normalsize
To enable the round-robin optimization, the order of the cluster optimization must be defined. It will generally be assumed in this paper that the following statement holds. 

\begin{assumption} \label{clusters}
The cluster optimization order $\Omega$ is a permutation of $\{1,\dots, C\}$.
\end{assumption}

The proposed algorithm begins by initializing the value of each state to zero. For each iteration $k$, the algorithm selects a cluster according to the optimization order, and the clustered Bellman operator $\Tb_{\Pi^k}^{c_k}$ is applied to the current value estimate $V_k(x)$. In the clustered Bellman operator, all clusters except for $c_k$ follow the fixed policy $\Pi^k_{- c_k}$, while the action for $c_k$ is optimized. The policy for $c_k$ is updated accordingly. A tie-breaking rule for actions is assumed for when $\argmax \Tb_{\Pi^k}^{c_k}V_{k}(x)$ has more than one element, such as choosing the $\alpha_{c_k}$ with the smallest index. 

\begin{algorithm}
\SetAlgoLined
 $V_{0}(x)\gets 0,\ \forall x\in\X$\;
 Initialize policy guess $\Pi^0(x) = \{\pi_{1}^0(x),\dots,\pi_{C}^0(x)\}, \forall x\in\X$\;
 Choose cluster optimization order $\Omega$\;
 $k=-1$\;
 
 \While{$ \|V_{k}-V_{k-1}\|_{\infty}>\epsilon$}{
  $k=k+1$\;
  $c_k\gets\Omega(k\text{ mod }C)$\;
  $V_{k+1}(x)\gets \Tb^{c_k}_{\Pi^k}V_k(x),\ \forall x\in\X$\;
  $\pi_{c_k}^{k+1}(x) \gets \argmax_{\alpha_{c_k}}\Tb^{c_k}_{\Pi^k}V_k(x),\ \forall x\in\X$\;
  $\Pi_{- c_k}^{k+1}(x) \gets \Pi_{- c_k}^k(x)$\;
 }
 \caption{Clustered Value Iteration (CVI)}
 \label{CVI}
\end{algorithm}

The cluster-based approach of CVI improves computation complexity versus standard VI as the optimization step searches over a significantly smaller space. Again, let each $\alpha_c$ have $M$ possible actions. Standard VI on cluster-based controls has a complexity of $O(M^C |\X|^2)$ per iteration; in comparison, CVI iterations have a complexity of $O(M |\X|^2)$, thus eliminating the per-iteration dependence on the number of clusters. While there will be an increase in the number of clusters, the savings per iteration will make up for the overall computation time as explored in simulation in Section \ref{sims}.

\subsection{Convergence}
This section studies the properties of the CVI algorithm. The first main theorem below establishes convergence.

\begin{theorem}\label{mainthm} \textbf{Convergence of CVI.} 
    Consider the CVI Algorithm (Algorithm \ref{CVI}) for a MDP such that Assumptions \ref{as: agent behavior}, \ref{as:reward}, and \ref{clusters} are fulfilled. Then, there exists, $\Vhat\in\mathbb{R}^{|\X|}$ such that $V_{k}\rightarrow\Vhat$ as $k\rightarrow\infty$.
\end{theorem}


The following lemmas will be used to prove Theorem \ref{mainthm}. The first result will show boundedness of the value estimates.
\begin{lemma}\label{l: bd} \textbf{Boundedness of Values.} 
The CVI iterates satisfy $\sup_{k\geq 0}\|V_{k}\|_{\infty}<\infty$.
\end{lemma}

\begin{proof} Let $r = \sup_{x\in\X,\abold\in\A} r(x,\abold)$.
    \ifarxiv
    \begin{align*}
        &V_{k+1}(x) \\
        &= \max_{\alpha_{c_k}} \sum_{x'}p(x'|x,\{\Pi^k_{- c_k}(x),\alpha_{c_k}\})(r(x,\{\Pi^k_{- c_k}(x),\alpha_{c_k}\}) + \gamma V_k(x')),\\
        &\leq r+\gamma \|V_k\|_{\infty}.
    \end{align*}
    \fi
    Therefore $\|V_k\|_{\infty}\leq r\sum_{i=0}^{k-1}\gamma^i + \gamma^k\|V_0\|_{\infty}$. As $\gamma<1$, $\limsup_{k\to\infty}\|V_k\|_{\infty}\leq r/(1-\gamma)$. 
\end{proof}

\begin{lemma}\label{l: mon}\textbf{Monotonicity of the Clustered Bellman Operator.} $\Tb_{\Pi}$ is monotone in the sense that if $V_1(x)\geq V_2(x)$ then $\Tb_{\Pi}V_1(x)\geq \Tb_{\Pi}V_2(x)$.
\end{lemma}
\begin{proof}
    Note that the action here is fixed. Result follows from \cite{bertsekas1995dynamic} Volume II, Lemma 1.1.1.
\end{proof}

Now we proceed with the proof of Theorem \ref{mainthm}.

\begin{proof}
    First, induction will show that the sequence $\{V_k(x)\}_{k\geq 0}$ increases monotonically for all $x$. The base case inequality $V_{0}(x)\leq V_{1}(x)$ will shown first. The estimated values are initialized to $V_{0}(x) = 0$. The next value is evaluated as $V_{1}(x)= \max_{\alpha_{c_0}}\sum_{x'\in\X} p(x'|x,\{\Pi^0_{- {c_0}}(x),\alpha_{c_0}\})(r(x,\{\Pi^0_{- {c_0}}(x),\alpha_{c_0}\}) + \gamma V_{0}(x'))= 
    \max_{\alpha_{c_0}}r(x,\{\Pi^0_{- {c_0}}(x),\alpha_{c_0}\})\geq V_{0}(x)$, thus holding the conclusion for any $c$ and $x$.
    
    Next we will show that $V_{k+1}(x)\geq V_k(x)$. Note that $\Tb_{\Pi_k}^{c_k}V_k(x) = V_{k+1}(x)$ yields new action $\abold_{k+1}$ for $x$. Letting $\Pi^{k+1} = \abold_{k_1}$ in Eq \eqref{eval} means that $\Tb_{\Pi^{k+1}}V_k(x) = V_{k+1}(x)$. 
    
    Consider  $V_{k+1}(x) = \Tb_{\Pi^k}^{c_k}V_k(x)$ with action maximizer $\alpha_{c_k}^{k+1}$ for $x$. If $\alpha_{c_k}^{k+1}= \alpha_{c_k}^k$, then  $\abold^{k+1} = \abold^k$ and
    $V_{k+1}(x) = \Tb_{\Pi^k}V_k(x)$. Since $V_k(s)  = \Tb_{\Pi^k}V_{k-1}(x)$, then by Lemma \ref{l: mon} $\Tb_{\Pi^k}V_k(x)\geq \Tb_{\Pi^k}V_{k-1}(x)$ for $V_k(x)\geq V_{k-1}(x)$ for all $x$. 
    
    However if $\alpha_{c_k}^{k+1}\neq \alpha_{c_k}^k$, then the value can be bounded:
    \ifarxiv
    \small{
    \begin{align*}
        &V_{k+1}(x) = \Tb_{\Pi^{k+1}}V_k(x) = \Tb_{\Pi^{k}}^{c_k}V_k(x), \\
        &=  \max_{\alpha_{c_k}\in\A_c}\sum_{x'\in\X}p(x'|x,\{\Pi_{- {c_k}}^k(x),\alpha_{c_k}\})(r(x,\{\Pi_{- {c_k}}^k(x),\alpha_{c_k}\})+\gamma V_k(x')),\\
        &> \sum_{x'\in\X}p(x'|x,\Pi^k(x))(r(x,\Pi^k(x))+\gamma V_k(x') )= \Tb_{\Pi^k}V_k(x).
    \end{align*}}
    \fi
    \normalsize
    Since $\Tb_{\Pi^k}V_k(x)\geq \Tb_{\Pi^k}V_{k-1}(x)$, then $ \Tb_{\Pi^{k+1}}V_k(x)\geq \Tb_{\Pi^k}V_{k-1}(x)$ for all $x$. By the boundedness property established by Lemma \ref{l: bd} and the above monotonicity of the sequence $\{V_k(x)\}_{k\geq 0}$, for all $x\in\mathcal{X}$, we conclude the convergence of $\{V_{k}\}$. 
\end{proof}

\begin{lemma} \textbf{Contraction.} The clustered Bellman operator satisfies, \label{contraction}

\small\begin{align}
    \max_{x\in\X}\vert \Tb^{c_1}_{\Pi_1}V_1(x)-\Tb^{c_2}_{\Pi_2}V_2(x)\vert \leq \gamma \max_{x\in\X}\vert V_1(x)-V_2(x)\vert.
\end{align}
\end{lemma}
\normalsize
Note that Lemma \ref{contraction} can be used to find a convergence rate for CVI, but it does not guarantee uniqueness of the fixed point $\lim_{k\to\infty} (\Tb^c_\Pi)^kV(x)$ unless $c_1=c_2$ and $\Pi_1=\Pi_2$ for all $k$. 

\begin{proof}
Note that:

\small\begin{align}
    \vert \max_x f(x,y,z)-\max_y g(x,y,z)\vert \leq \max_{x,y,z}\vert f(x,y,z)-g(x,y,z)\vert. \label{claim}
\end{align}

\normalsize 
Then,

\small
\ifarxiv
\begin{align}
        &\vert\Tb^{c_1}_{\Pi_1}V_1(x)-\Tb^{c_2}_{\Pi_2}V_2(x)\vert\nonumber\\
        \begin{split}
        =\Big\vert\max_{\alpha_1\in\A_1}\sum_{x'\in\X} p(x'|x,\{\Pi_{- c_1}(x),\alpha_{c_1}\})(r(x,\{\Pi_{- c_1}(x),\alpha_{c_1}\})+\gamma V_1(x'))\\
        -\max_{\alpha_2\in\A_2}\sum_{x'\in\X} p(x'|x,\{\Pi_{- c_2}(x),\alpha_{c_2}\})(r(x,\{\Pi_{- c_2}(x),\alpha_{c_2}\})+\gamma V_2(x'))\Big \vert
        \end{split}\nonumber\\
        &\leq \max_{\abold\in\A}\Big\vert\sum_{x'\in\X}p(x'|x,\abold)\Big(r(x,\abold)+\gamma V_1(x')\Big)-\sum_{x'\in\X}p(x'|x,\abold)\Big(r(x,\abold)+\gamma V_2(x')\Big)\Big\vert\label{use claim} \\
        &=\max_{\abold\in\A}\sum_{x'\in\X}p(x'|x,\abold)\Big\vert\gamma V_1(x') -\gamma V_2(x')\Big\vert\nonumber\\
        &\leq\max_{x\in\X}\gamma \Big\vert V_1(x') - V_2(x')\Big\vert\nonumber
\end{align}
\fi
\normalsize
where \eqref{use claim} uses \eqref{claim}.
\end{proof}

\subsection{Error Analysis}
In this section we discuss the performance of the policy output by CVI. While Theorem \ref{mainthm} establishes convergence, it does not guarantee uniqueness or optimality. In fact, the fixed point of the algorithm will depend on the initial policy guess and the cluster optimization order.

\begin{theorem} \label{th: fixed}
    \textbf{Consistency.} The true optimal value $V^*$ is a fixed point of $\Tb_{\Pi^*}^c$ for any $c\in\C$ and optimal policy $\Pi^*$.
\end{theorem}

\begin{proof}
    Let $\Pi^*(x) = \abold^*$ and $V^*(x)$ be a true optimal policy, optimal control for $x$, and optimal value of $x$ for $\mathcal{M}$. Now consider one iteration of CVI that optimizes with respect to cluster $c$. The next choice of $\alpha_c$ is $\alpha_c' = \argmax_{\alpha_c\in\A_c} \sum_{x'\in\X}p(x'|x,\{\Pi_{- c}^*(x),\alpha_c\})(r(x,\abold)+\gamma  V^*(x))$. If $\alpha_c'\neq \alpha_c^*$, then the choice $\abold' = \{\alpha_1^*,\dots,\alpha_c',\dots,\alpha_c^*\}$ would yield a higher value than $\abold^* = \{\alpha_1^*,\dots,\alpha_c^*,\dots,\alpha_c^*\}$. Therefore $\Pi^*(x)=\abold^*$ is not a true optimal policy and control. This contradicts the assumption that $V^*(x)$ is the optimal value and that $\Pi^*(x)=\abold^*$ are optimal policies and controls. Then $\Tb_{\Pi^*}^cV^*(x)=V^*(x)$. 
\end{proof}

While the consistency result shows that CVI is capable of finding the optimal answer, the next theorem provides a general performance bound that holds for every fixed point.

\begin{theorem}\label{bounds}
    \textbf{Suboptimality of Approximate Policy.} Consider the CVI Algorithm (Algorithm \ref{CVI}) for a MDP such that Assumptions \ref{as: agent behavior}, \ref{as:reward}, and \ref{clusters} are fulfilled. Let $V^*$ be the true optimal value under optimal policy $\Pi^*$, and let $\Vhat$ be a fixed point of CVI. Then,
    
    \small
\begin{align}
    \text{(UB)}\ \|V^*-\Vhat\|\leq \frac{1}{1-\gamma} \|\Vhat-\Tb_{\Pi^*}\Vhat\|\label{UB},\\
    \text{(LB)}\ \|V^*-\Vhat\|\geq \frac{1}{1+\gamma} \|\Vhat-\Tb_{\Pi^*}\Vhat\|.\label{LB}
\end{align}
\end{theorem}
\begin{proof}
\normalsize
\ifarxiv
\begin{align*}
    &\|V^*-\Vhat\|_{\infty} \\
    &= \|V^*-\Tb_{\Pi^*}\Vhat + \Tb_{\Pi^*}\Vhat - \Vhat\|_{\infty}\\
    &\leq \|\Tb_{\Pi^*}V^*-\Tb_{\Pi^*}\Vhat\|_{\infty} + \|\Tb_{\Pi^*}\Vhat-\Vhat\|_{\infty}\\
    &\leq \gamma\|V^*-\Vhat\|_{\infty} + \|\Tb_{\Pi^*}\Vhat-\Vhat\|_{\infty}\\
    &\Rightarrow (1-\gamma) \|V^*-\Vhat\|_{\infty} \leq \|\Tb_{\Pi^*}\Vhat-\Vhat\|_{\infty}\\
    &\Rightarrow\|V^*-\Vhat\|_{\infty} \leq \frac{1}{1-\gamma}\|\Tb_{\Pi^*}\Vhat-\Vhat\|_{\infty}
\end{align*}
\begin{align*}
    &\|\Vhat-V^*\|_{\infty} \\
    &= \|\Vhat-\Tb_{\Pi^*}\Vhat+\Tb_{\Pi^*}\Vhat-V^*\|_{\infty}\\
    &\geq \|\Vhat-\Tb_{\Pi^*}\Vhat\|_{\infty} - \|\Tb_{\Pi^*}\Vhat-\Tb_{\Pi^*}V^*\|_{\infty}\\
    &\geq \|\Vhat-\Tb_{\Pi^*}\Vhat\|_{\infty} - \gamma\|\Vhat-V^*\|_{\infty}\\
    &\Rightarrow (1+\gamma) \|V^*-\Vhat\|_{\infty} \geq \|\Tb_{\Pi^*}\Vhat-\Vhat\|_{\infty}\\
    &\Rightarrow\|V^*-\Vhat\|_{\infty} \geq \frac{1}{1+\gamma}\|\Tb_{\Pi^*}\Vhat-\Vhat\|_{\infty}
\end{align*}
\fi
\end{proof}

The right hand side of equations \eqref{UB}\eqref{LB} depend on the term $\|\Vhat - \Tb^*\Vhat\|$ which is a measure of the improvement under the full Bellman operator. Upon reaching a fixed point of CVI, a practitioner can either perform one evaluation of the full Bellman operator to check the bound, or use the fixed point of CVI as an initial guess for standard VI. Note that CVI can be considered as a type of approximate VI, and Theorem \ref{bounds} matches the general results of approximate VI \cite{bertsekas1995dynamic}.

\subsection{Hybrid Approach}
The structures of \eqref{UB} and \eqref{LB} imply that a hybrid approach may be used to find $V^*$ but with the computational savings of CVI; the error introduced by cheap CVI updates may be rectified with an occasional expensive VI update. This proposed method is outlined in Algorithm \ref{hybrid}. From any value guess $V_k$, CVI updates are used to compute $\Vhat_{k+1}$. The subsequent VI update corrects the policy guess, enabling the next iteration to use policy that escapes any local optimum found by CVI. 

Parameters $\delta$ and $\epsilon$ must be selected to run the fewest number of full Bellman updates as possible. There is a risk that the true Bellman operator will repeatedly output the \emph{same} policy; i.e. CVI had found the correct policy, but $\epsilon$ was not chosen small enough to satisfy the convergence condition for $\delta$. To avoid this scenario, we suggest choosing $\epsilon$ an order of magnitude less than $\delta$, and monitoring $\Pi_k$ versus $\Pi_{k+1}$ to check for improvement. An alternate convergence criteria is to use \eqref{UB} to use $\Vhat_{k+1}$ and $V_{k+1}$ to bound $\|\Vhat_{k+1}-V^*\|$. If the distance of $\Vhat$ to the optimum is sufficiently small, then the user may choose to terminate the Algorithm and not run any more evaluations of the full Bellman operator.

\begin{algorithm}
\SetAlgoLined
 Initialize values $V_{0}(x)\gets 0,\ \forall x\in\X$\;
 Initialize policy guess $\Pi^0(x), \forall x\in\X$\;
 $k=-1$\;
 \While{$\|V_{k}-V_{k-1}\|_{\infty}>\delta$}{
 $k=k+1$\;
 $\Vhat_{k+1}(x)\gets \text{CVI}(V_{k}(x), \Pi_k(x), \epsilon),\ \forall x\in \X$\;
 $V_{k+1}(x)\gets \Tb\Vhat_{k+1}(x),\ \forall x\in \X$\;
 $\Pi_{k+1}(x)\gets \argmax_{\abold} \Tb\Vhat_{k+1}(x)$\;
 }
 \caption{Hybrid CVI/VI}
 \label{hybrid}
\end{algorithm}

\begin{theorem}
\textbf{Convergence and Optimality of Hybrid CVI/VI.} Consider the Hybrid CVI/VI Algorithm (Algorithm \ref{hybrid}) for a MDP such that Assumptions \ref{as: agent behavior}, \ref{as:reward}, and \ref{clusters} are fulfilled.  Then $V_k\to V^*$ as $k\to\infty$.
\end{theorem}
\begin{proof}
Convergence may once again be shown through monotonicity and boundedness of updates. $V_k(x)\leq \Vhat_{k+1}(x)$ for all $x\in\X$ is already known for CVI updates. The next improvement $\Vhat_{k+1}(x) \leq V_{k+1}(x)$ for all $x\in\X$ may be determined as the $\Tb$ operator searches over the entire action space and will perform at least as well as a CVI update.

To see optimality, note that for any evaluation $\Tb \Vhat_k(x)=V_k(x)$, if $\Vhat_k(x)=V_k(x)$ then the fixed point $V_k(x)=V^*(x)$ of $\Tb$ has been achieved; else there will be monotonic improvement $V_k(x)\geq \Vhat_k(x)$. As $V^*(x)\geq \Vhat_k$ for all $k$ (due to $r(x,\abold)\geq 0$ and $V_0(x)=0$), $V_k(x)\to V^*(x)$ as $k\to\infty$.
\end{proof}

\ifarxiv
\subsection{Extensions \& Variants}
The Bellman operator may be substituted with the clustered Bellman operator in a variety of dynamic programming approaches given the clustered setting. In this section, we recap a few of the most immediate extensions and suggest how the clustered Bellman operator could be utilized.

\paragraph{Value Function Approximation (VFA)} Recall that CVI was motivated to handle combinatorial action spaces. This may be combined with VFA \cite{bertsekas1995dynamic}, which handles large state spaces by approximating the value function with a parametric class. A cluster-based formulation is well-suited for parametric approximation, as each features may be defined per cluster.

\paragraph{Clustered Policy Iteration} The PI algorithm may also make use of CVI updates to reduce computation in the policy improvement step. While monotonicity will still give convergence, the computation savings may not be significant. Recall that the per-iteration complexity of PI is $O(|\mathcal{A}||\X|^2+|\X|^3)$; the computation here will likely be dominated by the matrix inversion term $|\X|^3$, meaning that the clustered Bellman operator may not be the most impactful step for speed. 

\paragraph{Approximate PI (API)} In practice, API \cite{bertsekas1995dynamic} is often used to compute policy evaluation and improvement up to some allowed error bounds. The results of Theorem \ref{bounds} are similar to those for API as both techniques perform an approximate policy improvement step, but API has an additional step of approximately evaluating  proposed policies. In CVI, policy evaluation is only performed after the found policy converges. The two techniques are similar in concept, but CVI is based on structure inherent to a cluster-based control policy. 


\paragraph{State Aggregation} Clustered policies are thematically related to state aggregation methods, but the technical formulations are distinct. Both methods aim to reduce the dimensionality of the problem by solving smaller MDPs. In state aggregation, similar states are grouped together to form a new system with a reduced state space and weighted transitions. In comparison, the clustered formulation works for applications with partitioned substates and factorable transitions. These two methods could be used in tandem to enjoy further dimensionality reduction, whereby a clustered formulation could use aggregation techniques on the set of substates. 
\fi

\section{Separable Systems} \label{separable section}
The next section of this paper will focus on a special class of TI-MDPs that have additive reward functions, also known as factored MDPs. It will be shown in this class the reward structure guarantees that CVI will find optimal policies.

\begin{definition}
The TI-MDP is called a \emph{separable} system if the rewards are additive with respect to the cluster assignments: 

\small\begin{equation}
    r(x,\abold) = \sum_{c=1}^C r_c(x_c,\alpha_c) 
\end{equation}
\end{definition}
\normalsize
A more general version of this definition is additive rewards with respect to the agents, i.e. $r(x,\abold) = \sum_{n=1}^Nr_n(x_n,\alpha_n)$. The rest of this section will assume that the system is a separable TI-MDP, which arises in clustered control applications. For example, agents may fall into a customer type, and the central planner may define a reward based on the type, i.e. $r_c(x_c,\alpha_c) = \sum_{n\in \C_c}\mathbbm{1}(x_n \in \mathcal{D}_c)$ where $\mathcal{D}_c$ is a set of desirable actions for agents of type $c$, or the reward may represent the proportion of agents representing general good behavior, i.e. $r(x,\abold) = \frac{1}{N}\sum_{n\in\N} \mathbbm{1}(x_n\in \mathcal{D}_n)$. Reward functions that are not separable depend on the whole state. For example, zero-sum games, the prisoner's dilemma, and chicken cannot be represented with separable reward functions, as the reward depends on the states of both players. 


\begin{lemma} \label{separability}
    \textbf{Separability of Bellman Operator on Separable Systems.} Consider solving a separable TI-MDP with VI. At some iteration $k$, the full Bellman operator can be expressed as the sum of $C$ modified Bellman operators that each maximize over the action space for one cluster.
    
    \small\begin{align}
        V_{k+1}(x)=\Tb V_k(x) &= \sum_{c=1}^C\Tb^c V_k(x) = \sum_{c=1}^C V_{k+1}(x_c),
    \end{align}
    \normalsize
    where the local cluster values are defined iteratively as $V_0(x_c)=0$ and,
    \small\begin{align}
        \Tb^cV_k(x)=\max_{\alpha_c\in\A_c} \mathbb{E}\left [r_c(x_c,\alpha_c) + \gamma V_{k}(x_c') \Big\vert x_c'|x,\alpha_c\sim T_c \right], \label{local values}
    \end{align}
    \normalsize
    where $T_c=p(x_c'|x,\alpha_c)$ is one of the factors of the factored transition matrix.
\end{lemma}

Note that any evaluation of the clustered Bellman operator requires a full state $x$ for proper conditioning; thus, the local cluster value equation \eqref{local values} is a function of $x$. 

\begin{proof}
    The proof will show by induction that each iteration of standard VI can be split into independent optimizations over the action space for each cluster. 
    
    Base case. Initialize $V_0(x)=0$ for all $x$. Then $V_1(x) = r(x,\abold)=\sum_{c=1}^C r_c(x_c,\alpha_c)=\sum_{c=1}^CV_1(x_c)$. 
    
    Then, 
    \ifarxiv
    \small\begin{align*}
        &V_2(x)=\Tb V_1(x) \\
        &=\max_{\abold\in\A}\sum_{x'\in\X}p(s_1'|x,\alpha_1)\dots p(x_C'|x,\alpha_C)\left [\sum_{c=1}^C r_c(x_c,\alpha_c) + \gamma r(x',\abold^{k=1})\right ]\\
        &= \max_{\abold\in\A}\Big [\sum_{x_1'\in\X_1} p(x_1'|x,\alpha_1)\Big [\sum_{x_2'\in\X_2}p(x_2'|x,\alpha_2)\Big[ \dots  \sum_{x_C'\in\X_C}p(x_C'|x,\alpha_C)\Big [\sum_{c=1}^Cr_c(x_c,\alpha_c) + \gamma  r(s', \abold^{k=1})\Big ]\Big] \Big] \Big]\\
        &= \max_{\abold} \mathbb{E}_{x_1'|x,\alpha_1}\dots\mathbb{E}_{x_C'|x,\alpha_C}\left[\sum_{c=1}^Cr_c(x_c,\alpha_c) + \gamma r(x', \abold^{k=1})\right]\\
        &= \max_{\abold} \mathbb{E}_{x_1'|x,\alpha_1}\dots\mathbb{E}_{x_C'|x,\alpha_C}\left[\sum_{c=1}^Cr_c(x_c,\alpha_c) + \gamma \sum_{c=1}^Cr_c(x_c', \alpha^{k=1}_{c})\right]\\
        &= \max_{\abold}\sum_{c=1}^C \left [ \mathbb{E}_{x_1'|x,\alpha_1}\dots\mathbb{E}_{x_C'|x,\alpha_C}\left[r_c(x_c,\alpha_c) + \gamma r_c(x_c',\alpha^{k=1}_c)\right]\right ]\label{pre}
        \end{align*}
        \begin{align*}
        &= \max_{\abold}\sum_{c=1}^C \mathbb{E}_{x_c'|x,\alpha_c}\left[r_c(x_c,\alpha_c) + \gamma r_c(x_c',\alpha^{k=1}_c)\right]\\
        &= \sum_{c=1}^C \max_{\alpha_c} \mathbb{E}_{x_c'|x,\alpha_c}\left[r_c(x_c,\alpha_c) + \gamma r_c(x_c',\alpha^{k=1}_c)\right] \\
        &= \sum_{c=1}^C \max_{\alpha_c} \mathbb{E}_{x_c'|x,\alpha_c}\left[r_c(x_c,\alpha_c) + \gamma V_1(x_c')\right]\\
        &=\sum_{c=1}^C \Tb^cV_1(x)=\sum_{c=1}^C V_2(x_c)=V_2(x)
    \end{align*}
    \fi
    
    \normalsize Notice that $r_c(x_c')$ is independent of $x_{c'}'$ where $c'\neq c$. The last line means that the resulting $V$ is decomposable over $c$, and that each element of $\abold$ may be optimized independently. 
    
    \ifarxiv
    General case. 
    \small{\begin{align*}
        &\Tb V_k(x) \\
        &=  \max_{\abold\in\A}\sum_{x'}p(x_1'|x,\alpha_1)\dots p(x_C'|x,\alpha_C)\left[ \sum_{c=1}^Cr_c(x_c,\alpha_c) + \gamma V_k(x')\right]\\
        &= \max_{\abold\in\A} \mathbb{E}_{x_1'|x,\alpha_1}\dots\mathbb{E}_{x_C'|x,\alpha_C}\left[\sum_{c=1}^Cr_c(x_c,\alpha_c) + \gamma \sum_{c=1}^CV_{k}(x_c')\right]\\
        &= \max_{\abold\in\A}\sum_{c=1}^C  \mathbb{E}_{x_1'|x,\alpha_1}\dots\mathbb{E}_{x_C'|x,\alpha_C}\left[r_c(x_c,\alpha_c) + \gamma V_{k}(x_c')\right]\\
        &= \max_{\abold\in\A}\sum_{c=1}^C \mathbb{E}_{x_c'|x,\alpha_c}\left [r_c(x_c,\alpha_c) + \gamma V_{k}(x_c')\right]\\
        &= \sum_{c=1}^C \max_{\alpha_c\in\A_c} \mathbb{E}_{x_c'|x,\alpha_c}\left [r_c(x_c,\alpha_c) + \gamma V_{k}(x_c')\right ] = \sum_{c=1}^C \Tb ^cV_{k}(x) = \sum_{c=1}^CV_{k+1}(x_c)=V_{k+1}(x).
    \end{align*}}
    \fi
    
\normalsize The last line shows that $V$ may again be decomposed for general $k$ and that $\alpha_c$ may be optimized independently.
\end{proof}

Lemma \ref{separability} shows that each Bellman update may be written as a sum over clusters, where each addend is only dependent on one element of the control vector, $\alpha_c$. With this decomposition, the next theorem establishes that CVI applied to separable systems will converge to the true optimal value. 
\begin{theorem}\label{optimality}
    \textbf{Optimality of CVI on Separable Systems.} Consider the CVI Algorithm for a separable TI-MDP. The CVI algorithm will converge to the true optimal value, $\Vhat = V^*$. 
\end{theorem}

\begin{proof}
First it will be shown that local updates on separable systems as in \eqref{local values} converge to the unique optimal value, and then it will be shown that \eqref{local values} performs the same optimization as CVI. As shown in the proof of Lemma \ref{bellman}, a VI update for a separable system can be written as $\Tb V_k(x)=\sum_{c=1}^C \Tb^c V_k(x)$. Therefore $ V^*(x) = \lim_{k\to\infty} \Tb V_k(x)=\lim_{k\to\infty} \Tb^k V_0(x) = \lim_{k\to\infty} \sum_{c=1}^C(\Tb^{c})^kV_0(x) = \sum_{c=1}^CV^{*}(x_c)$.
\small
\ifarxiv
\begin{equation*}
    V^*(x) = \lim_{k\to\infty} \Tb V_k(x)=\lim_{k\to\infty} \Tb^k V_0(x) = \lim_{k\to\infty} \sum_{c=1}^C(\Tb^{c})^kV_0(x) = \sum_{c=1}^CV^{*}(x_c).
\end{equation*}
\fi

\normalsize
Next, it will be shown that \eqref{local values} performs the same optimization as CVI. Consider a sequence of local cluster values $\{U(x_c)\}_{k\geq 0}$ for some cluster $c$, except now instead of being optimized at each time step as in \eqref{local values}, it will be optimized one time out of every $C$ time steps as in CVI. Let us construct $\{U(x_c)\}_{k\geq 0}$. Let $b_c$ be the index of cluster $c$ in $\Omega$; i.e. $c$ is optimized whenever $k\text{ mod }C = b_c$. Then, the value of cluster $c$ at time $b_c$ is $U_{b_c-1}(x_c) = \Tb_{\pi}\dots\Tb_{\pi} U_0(x)$ which is $b_c-1$ evaluations of the Bellman operator with initial policy $\pi$. 

Next, define the operator $\overline{\Tb}^cU(x) = \Tb_{\pi}\dots \Tb_{\pi }\Tb^cU(x)$ where $\Tb^c$ preforms the maximization step and produces candidate $\alpha_c$ that is then evaluated $C-1$ times on the resulting cluster value. Note that $\overline{\Tb}^c$ is a contraction (with parameter $\gamma^C$), so taking the limit of $k\to\infty$ of $(\overline{\Tb}^c)^kU(x)$ will converge to a unique fixed point. Additional evaluations of $\Tb_{\pi}$ will not change the final value, and $b-1$ applications of $\Tb_{\pi}$ to $U_0^c$ (with arbitrary $\pi$) before $\overline{\Tb}^c$ will change the initial guess to $U^c_{b_c-1}$, but will not change the fixed point. Thus $\lim_{k\to\infty} \sum_{c=1}^C (\overline{\Tb}^c)^k(\Tb_{\pi})^{b_c-1}U_0(x)=\lim_{k\to\infty} \sum_{c=1}^C(\overline{\Tb}^c)^kU_{b_c-1}(x)=\sum_{c=1}^CU^{*}(x_c) = \sum_{c=1}^CV^{*}(x_c).$

\small
\ifarxiv
\begin{equation*}
    \lim_{k\to\infty} \sum_{c=1}^C (\overline{\Tb}^c)^k(\Tb_{\pi})^{b_c-1}U_0(x)=\lim_{k\to\infty} \sum_{c=1}^C(\overline{\Tb}^c)^kU_{b_c-1}(x)=\sum_{c=1}^CU^{*}(x_c) = \sum_{c=1}^CV^{*}(x_c).
\end{equation*}
\fi
\end{proof}

\normalsize
For separable TI-MDPs, the user may enjoy both the value improvement from the clustered policy and the complexity savings of CVI without sacrificing performance.


\section{Choosing Cluster Assignments}\label{cluster section}
It has previously been assumed that the cluster assignments were known \emph{a priori}. This is a reasonable assumption when classes of agents are apparent or existing; i.e. category of customer. For general systems, however, the clustered control policy may be used to increase the received value without relying on intuition about existing agent classes. The objective in this section is to find clustering assignments that yield good value when none are suggested.

\begin{definition}
An \emph{optimal $C$ clustering assignment} for a TI-MDP is $\mathcal{C}^*$ such that the resulting optimal value is maximal;

\small\begin{gather} \label{optcluster}
    \C^*\in\argmax_{\C}\ V^*(x)\\
    \text{s.t. }\ |\C|=C,\ \C_1\cup\dots\cup\C_C =\N,\ \C_u\cap\C_v=\emptyset\ \nonumber\\
    \forall u\neq v,\  u,v\in[1,C].\nonumber
\end{gather}
\end{definition}

\normalsize
Note that \eqref{optcluster} requires $O(C^N)$ evaluations if solved by enumeration, and is known to be NP-hard \cite{calinescu2007maximizing}. In this section, an approximate solution to the optimization problem \eqref{optcluster} will be explored using submodularity. Let $V(\C_k)$ mean the vector of values associated with clustering assignment $\C$ that consists of $k$ clusters. Let $\C^*_k$ be the optimal clustering assignment of $k$ clusters, and let $\widehat{\C}_k$ be the clustering assignment of $k$ clusters as found by an approximate algorithm. All following discussion will be for separable TI-MDPs.

The first definition is for submodular functions, often used in discrete optimization, and it means the benefit of adding elements to a set decreases as the set grows in size. Maximizing a submodular function means that the possible improvement between the $k$ and $k+1$ indices is bounded. 

\begin{definition}
Given a finite set $E$, a \emph{submodular} function $f: 2^E\to\R$ is such that for every $A\subseteq B\subseteq E$ and $i\in E \setminus B$,

\small\begin{equation}
    f(A\cup \{i\})-f(A)\geq f(B\cup \{i\})-f(B).
\end{equation}
\normalsize
Furthermore, $f$ is \emph{monotone} if $f(B)\geq f(A)$.
\end{definition}

\subsection{Submodularity of Value Functions}
We first show that value function for certain TI-MDPs is a submodular function. With the following result, the optimization problem \eqref{optcluster} may be approached with techniques developed for submodular function.

\begin{lemma}
\textbf{Submodularity of TI-MDP.} Consider an agent-factorable TI-MDP with $r(x,\abold)$ submodular with respect to the number of agents. Then the value function $V(x)$ is also submodular with respect to the number of agents. I.e., for sets of agents $A\subseteq B\subseteq \N$ and agent $n\in \N\setminus B$, \label{submodular value}

\small\begin{equation}
    V^{A\cup\{n\}}(x)-V^A(x)\geq V^{B\cup \{n\}}(x)-V^B(x), \label{submdp}
\end{equation}
\normalsize
where $V^D(x)$ is the value function of the MDP comprised of the agents in set $D$. 
\end{lemma}

\begin{proof}
$V(s) = \lim_{k\to\infty} \Tb^k V_{0}(x)$, so it suffices to show that \eqref{submdp} holds for any $k$th application of the Bellman operator. 
\ifarxiv
\small\begin{align*}
    V_1^A(x) = r(x_A,\abold);\ V_1^{A\cup \n}=r(x_{A\cup \n},\abold);\\
    V_1^B(x) = r(x_B,\abold);\ V_1^{B\cup \n}=r(x_{B\cup \n},\abold).
\end{align*}
\fi
\normalsize
As $r(x,\abold)$ is given to be submodular, the claim will hold for $k=1$. Now show for general $k$.

\small
\ifarxiv
\begin{align*}
    &\Tb V_k^{B\cup \n}(x) - \Tb V_k^B(x)\\
        &= \sum_{x_{B\cup \n}'}p(x_{B\cup \n}'|x,\abold)[r(x_{B\cup \n},\abold)+\gamma V^{B\cup \n}_{k-1}(x')]-\sum_{x_{B}'}p(x_{B}'|x,\abold)[r(x_{B},\abold)+\gamma V^B_{k-1}(x')]\\
        &=\sum_{x_B'}p(s_B'|x,\abold)\Big[\sum_{x_{\n}'}p(x_{\n}'|x,\abold)[r(x_{B\cup \n},\abold)+\gamma V^{B\cup \n}_{k-1}(x')] - [r(x_{B},\abold)+\gamma V^B_{k-1}(x')]\Big]\\
        &=\sum_{x_B'}p(x_B'|x,\abold)\Big[\sum_{x_{\n}'}p(x_{\n}'|x,\abold)\Big[r(x_{B\cup \n},\abold)-r(x_{B},\abold)+\gamma \Big(V^{B\cup \n}_{k-1}(x')- V^B_{k-1}(x')\Big)\Big]\Big]\\
        &\leq\sum_{x_B'}p(x_B'|x,\abold)\Big[\sum_{x_{\n}'}p(x_{\n}'|x,\alpha)\Big[r(x_{A\cup \n},\abold)-r(x_{A},\abold)+\gamma \Big(V^{A\cup \n}_{k-1}(x')- V^A_{k-1}(x')\Big)\Big]\Big]\\
        &=\sum_{x_A'}p(x_A'|x,\abold)\sum_{x_{B\setminus A}'}p(x_{B\setminus A}'|x,\abold)\Big[\sum_{x_{\n}'}p(x_{\n}'|x,\abold)\Big[r(x_{A\cup \n},\abold)-r(x_{A},\abold)+\gamma \Big(V^{A\cup \{n\}}_{k-1}(x')- V^A_{k-1}(x')\Big)\Big]\Big]\\
        &=\sum_{x_A'}p(x_A'|x,\abold)\Big[\sum_{x_{\n}'}p(x_{\n}'|x,\abold)\Big[r(x_{A\cup \n},\abold)-r(x_{A},\abold)+\gamma \Big(V^{A\cup \n}_{k-1}(x')- V^A_{k-1}(x')\Big)\Big]\Big]\\
    &=\Tb V_k^{A\cup \n}(x) - \Tb V_k^A(x)
\end{align*}
\fi
\normalsize
where the inequality holds due to submodularity of the reward function and the induction hypothesis. By induction, $\Tb V_k^{A\cup \n}(x) - \Tb V_k^A(x)\geq \Tb V_k^{B\cup \n}(x) - \Tb V_k^B(x)$ for all $k$, so take the limit as $k\to\infty$.
\end{proof}

The reward function for an agent-separable TI-MDP can be written in the form $r(x,\abold) = \sum_{n}r_n(x_n,\alpha_n)$. Linear functions are modular (submodular), so $r(x,\abold)$ is submodular and therefore will satisfy the requirements for Lemma \ref{submodular value}.

\begin{lemma}\label{monotonicity}\textbf{Monotonicity of TI-MDPs}: Consider a TI-MDP with $r(x,\abold)$ monotone with respect to the number of agents. Then the value function $V(x)$ is also monotone with respect to the number of agents.
\end{lemma}
\begin{proof} 
\ifarxiv
This again can be shown via induction.
\begin{align*}
    V_1^B(x)-V_1^A(x) = r(x_B,\abold)-r(x_A,\abold)\geq 0
\end{align*}
which holds by the given monotonicity.

Induction on $k$:
\small
\begin{align*}
    &\Tb V_k^B(x)-\Tb V_k^A(x) \\
    &=\sum_{x'_A}p(x'_B|x,\abold)\Big[\sum_{x'_{B\setminus A}}p(x'_{B\setminus A}|x,\abold)\Big[r(x_B,\abold)+\gamma V_{k-1}^B(x') \Big]-r(x_A,\abold)+\gamma V_{k-1}^A(x')\Big]\\
    &=\sum_{x'_A}p(x'_B|x,\abold)\sum_{x'_{B\setminus  A}}p(x'_{B\setminus A}|x,\abold)\Big[r(x_B,\abold)-r(x_A,\abold)+\gamma (V_{k-1}^B(x') - V_{k-1}^A(x'))\Big]\geq 0\\
\end{align*}
By induction, $\Tb V_k^B\geq \Tb V_k^A$ for all $k$, so therefore $\lim_{k\to\infty} V_k^B(x)\geq \lim_{k\to\infty} V_k^A(x)$ and $V^B(x)\geq V^A(x)$.

\fi
\normalsize
\end{proof}

\subsection{Greedy Clustering}
With the notions of submodularity established, we return to the optimal clustering problem. The optimization problem can be restated as, 

\small\begin{gather}
    \max_{\mathcal{C}} \sum_{c=1}^C V^{*}(x_c)\label{optclusters},\\
    \text{s.t. }\ |\C|=C,\ \C_1\cup\dots\cup\C_C =\N,\ \C_u\cap\C_v=\emptyset\ \nonumber\\
    \forall u\neq v,\  u,v\in[1,C].\nonumber
\end{gather}
\normalsize
As discussed in the introduction, \eqref{optclusters} can be represented as an optimization problem subject to a partition matroid constraint \cite{calinescu2007maximizing}. While that method provides a close approximation to the optimal, it was not formulated for MDPs and requires a complicated implementation for approximation subroutines. In this section we will explore a simple construction that builds upon the optimality of CVI queries on separable systems. 

In \cite{zhao2005greedy}, Zhao et. al. suggest a Greedy Splitting Algorithm (GSA) to solve \eqref{optclusters} for minimization that achieves an approximation $f(\hC)\leq (2-\frac{2}{k})f(\C^*)$ for monotonone submodular $f$ in $O(kN^3\theta)$ time where $\theta$ is the time to query $f$. This approach begins with all the agents as one cluster, and then takes the best possible refinement of $k-1$ existing clusters to propose $k$ clusters. Note that this method only requires computation of $V$, which is a perfect use case for CVI. Adapting GSA to a maximization problem, however, loses the error bound guarantees due to non-negativity constraints. In general, the maximization of submodular functions subject to a cardinality constraint is harder than minimization \cite{calinescu2007maximizing} and includes the maximum coverage problem.
 
Nevertheless, we propose adapting GSA for value maximization of cluster assignments because it guarantees value improvement at each iteration and it only requires calls to a CVI subprocess. A maximization version of GSA is presented in Algorithm \ref{GSAR} and it will be shown that the algorithm finds clustering assignments whose values form a submodular lower bound of the optimal values of \eqref{optclusters}. 

In general, $C$ will refer to the desired final number of clusters and $k$ will be the number clusters at an intermediate step. The initialization defines one cluster of all the agents. At each subsequent query for new clusters, the algorithm searches for some existing cluster $U$ and a split $\{U- X, X\}$ that provides the most value improvement. This is repeated until $C$ clusters are achieved and the final value is returned. The number of possible sets formed by splitting $n$ elements split into $k$ disjoint subsets may found by the Stirling number of the second kind $S(n,k)$. Therefore the number of possible splits across some clustering assignment is $\sum_{c=1}^CS(|\C_c|,2)$.

\begin{algorithm}
\SetAlgoLined
 $\hC_1 = \{\N\}$\;
 $\widehat{V}_1 = \text{CVI}(\hC_1)$\;
 \For{$k\in\{2,\dots,C\}$}{
 $(X_{k}, U_{k-1})\gets \text{argmax} \{V(X)+V(U- X) -V(U)\ |\ \emptyset\subset X\subset U,\ U\in \hC_{k-1}\}$\;
 $\hC_{k} \gets \{\hC_{k-1}- U_{k-1}\}\cup \{X_k, U_{k-1}- X_k\}$\;
 $\widehat{V}_k = \text{CVI}(\hC_k)$\;}
 \caption{Greedy Splitting Algorithm - Reward}
 \label{GSAR}
\end{algorithm}

\begin{remark} The optimal values for an optimal $k$-clustering is lower bounded by the optimal value for a clustering found by GSA-R: $V^*(\C^*_k)\geq V^*(\hC_k)$.

The clustering assignments found by GSA-R are optimal for $k=1, 2, N$, i.e.  $\hC_1=\C^*_1, \hC_2=\C^*_2, \hC_N=\C^*_N$. For $k=2$, this can be seen as optimizing over the set of 2-cluster assignments is equivalent to optimizing over one split.
\end{remark}

\begin{lemma}\label{splitting}
\textbf{Monotonicity of Splitting.} Consider a  TI-MDP with a $k$-clustering assignment. The optimal value $V^*(\C_{k+1})$ will at least as much as $V^*(\C_k)$ if the $k+1$ cluster is formed by splitting one of the existing $k$ clusters. 
\end{lemma}

\begin{proof}
\ifarxiv
Proof by induction.

Base case: Assume there are $k=2$ clusters, $c_1$ and $c_2$.  The value of a state under $\pi^*$ is,
\begin{equation*}
    V^*_{k=2}(x) = \max_{
    \alpha_1, \alpha_2}\sum_{x'\in\X} p(x'|x,\{\alpha_{c_1},\alpha_{c_2}\}) (r(x,\{\alpha_{c_1},\alpha_{c_2}\})+\gamma V^*(x')).
\end{equation*}
If $a_{c_1} = a_{c_2}$, then this is equivalent to the one-cluster case and the same value is recovered $V^*_{k=2}(x) =V^*_{k=1}(x)$.

Else if $a_{c_1} \neq a_{c_2}$ then,
\begin{align*}
    &V^*_{k=2}(x) = \max_{\alpha_{c_1}, \alpha_{c_2}}\sum_{x'\in\X} p(x'|x,\{\alpha_{c_1},\alpha_{c_2}\})(r(x,\{\alpha_{c_1},\alpha_{c_2}\})+\gamma V^*(x')),\\
    &\geq  \max_{\alpha}\sum_{x'} p(x'|x,\{\alpha,\alpha\})(r(x,\{\alpha,\alpha\})+\gamma V^*(x'))=V^*_{k=1}(x).
\end{align*}
Inductive step: Say there are $k'$ clusters and that the controller uses policy $\Pi^*$. Take cluster $c'$ and split it into two clusters, $\{c_1', c_2'\}$ where $c_1'\cup c_2'=\emptyset$ and $c_1'\cap c_2' = c'$. This forms $k'+1$ clusters. Again if $\alpha_{c_1'}=\alpha_{c'}$ and $\alpha_{c_2'}=\alpha_{c'}$, then we recover the same value $V^*_{k'}(x)=V^*_{k'+1}(x)$. However if $\alpha_{c_1'}\neq\alpha_{c'}$ or $\alpha_{c_2'}\neq\alpha_{c'}$, then,
\begin{align*}
    &V^*_{k'+1}(x)\\
    &=\max_{\alpha_1,\dots,\alpha_{c_1'},\alpha_{c_2'}}\sum_{x'} p(x'|x,\{\alpha_1,\dots,\alpha_{c_1'},\alpha_{c_2'}\}) \times( r(x,\{\alpha_1,\dots,\alpha_{c_1'},\alpha_{c_2'}\})+\gamma V^*(x')),\\
    &\geq \max_{\alpha_1,\dots,\alpha_{k'},\alpha_{k'}}\sum_{x'} p(x'|x,\{\alpha_1,\dots,\alpha_{k'},\alpha_{k'}\}) \times (r(x,\{\alpha_1,\dots,\alpha_{k'},\alpha_{k'}\})+\gamma V^*(x')),\\
    &=V^*_{C'}(x).
\end{align*}
\fi
\end{proof}
\normalsize
This result on the splittings further shows that the resulting values for the optimal clustering problem \eqref{optclusters} monotonically improve for increasing $k$.

\begin{lemma} \textbf{Monotonicity of Optimal Clusterings.} Consider a TI-MDP. The series of optimal values for each optimal clustering assignment is a monotonically increasing sequence: \label{monotone clusters}

\small\begin{align}
    V^*(\C^*_1)\leq \dots\leq V^*(\C^*_k)\leq \dots\leq V^*(\C^*_C).
\end{align}
\end{lemma}

\normalsize
\begin{proof}
By Lemma \ref{splitting}, $V^*(\C_k^*)\leq V^*(\C_{k+1})$ where $\C_{k+1}$ is a split of $\C_k^*$. Then $V^*(\C_{k+1})\leq V^*(\C_{k+1}^*)$.
\end{proof}

Note that Lemmas \ref{splitting} and \ref{monotone clusters} hold for non-separable TI-MDPs. For the separable case, we can find better performance guarantees in that the values have diminishing returns for increased numbers of clusters.      

\begin{theorem}\label{submodular} \textbf{Performance of GSA-R}: Consider a TI-MDP with $r(x,\abold)$ submodular with respect to the number of agents. The sequence of $\hC_k$ found by GSA-R satisfies,

\small\begin{align*}
    V^*(\hC_{k+1})-V^*(\hC_{k})\leq V^*(\hC_{k})-V^*(\hC_{k-1}).
\end{align*}
\end{theorem}
\normalsize
\begin{proof} In this proof, all value functions are assumed to be evaluated for the optimal policy.

\emph{Base Case:} We need to show $V(\hC_2)-V(\hC_1)\leq V(\hC_1)-V(\hC_0)$, but for zero clusters $V(\hC_0) = 0$, so we need to show $V(\hC_2)-V(\hC_1)\leq V(\hC_1)$. Denote the set $\hC_2=\{U, W\}$. By Lemma \ref{monotonicity}, $V(\hC_1)\geq V(U)$ and $V(\hC_1)\geq V(W)$. Thus $V(\hC_2) = V(U)+V(W) \leq 2V(\hC_1)$ as desired. 

\emph{Induction:} The goal is to show that $V(\hC_{k+1})-V(\hC_{k})\leq V(\hC_k)-V(\hC_{k-1})$. Let $U_k\in\hC_k$ be the cluster that is split into $X_{k+1}$ and $Y_{k+1}$ to form $\hC_{k+1}$, i.e. $X_{k+1}\cup Y_{k+1} = U_k$, $X_{k+1}\cap Y_{k+1} = \emptyset$. Then either $U_k\subset U_{k-1}$ or $U_k\not\subset U_{k-1}$.

\small\begin{align*}
    V(\hC_{k+1})-V(\hC_k) = V(X_{k+1}) + V(Y_{k+1}) - V(U_{k})
\end{align*}
\normalsize
If $U_k\not\subset U_{k-1}$ then optimally splitting $U_{k-1}$ provided a bigger value increase than optimally splitting $U_k$. This is because both sets were available to select to split at stage $k-1$, but the greedy algorithm selected $U_{k-1}$. Thus $V(X_k)+V(Y_k)-V(U_{k-1})\geq V(X_{k+1})+V(Y_{k+1})-V(U_k)$ and therefore $V(\hC_k)-V(\hC_{k-1})\geq V(\hC_{k+1})-V(\hC_k)$.

In the other case, $U_{k-1}$ is split, and then $X_k$ or $Y_k$ are selected as the next $U_k$. Say that $X_k=U_{k-1}$. 
Need to show:

\small\begin{align}
    &V(X_{k+1})+V(Y_{k+1})-V(X_{k+1}\cup Y_{k+1})\nonumber\\
    &\leq V(X_k)+V(Y_k)-V(X_k\cup Y_k)\nonumber\\
    &=V(X_{k+1}\cup Y_{k+1}) + V(Y_k)- V(X_{k+1}\cup Y_{k+1}\cup Y_k)\label{star}
\end{align}
\normalsize
By submodularity, $V(X_{k+1} \cup Y_{k+1})-V(X_{k+1}) -V(Y_{k+1}) \geq V(X_{k+1} \cup Y_{k+1} \cup Y_k) - V(X_{k+1}\cup Y_k) - V(Y_{k+1})$. Therefore to claim \eqref{star}, we need $V(X_{k+1}\cup Y_{k+1})+V(Y_k)\geq V(X_{k+1}\cup Y_k)+V(Y_{k+1})$. 

Recall that $X_{k+1}\cup Y_{k+1}\cup Y_k = U_{k-1}$. The above equation shows the value of two possible splittings of $U_{k-1}$: the LHS is the split $\{X_{k+1}\cup Y_{k+1}, Y_k\}$ and the RHS is the split $\{X_{k+1}\cup Y_k, Y_{k+1}\}$. As the algorithm selects the best splits in a greedy fashion, the LHS is a better split and thus has a greater value. Thus \eqref{star} holds, and so does the original claim. 
\end{proof}


\begin{remark}
This diminishing returns property may be used as a stopping criterion when a desired final number of clusters is not known \emph{a priori}. For example, consider a clustering $\hC_k$ with known $V^*(\hC_{k})$. The clusterings at the previous iterations are known, so  $V^*(\hC_{k-1})$ is known and thus $V^*(\hC_{k})-V^*(\hC_{k-1})=\delta_k$ is known. By Theorem \ref{submodular}, $\delta_k \geq \delta_{k+1}$. Therefore, if $\delta_k$ is sufficiently small, the user may decide to terminate the algorithm and not compute $\hC_{k+1}$.
\end{remark}

\section{Examples} \label{sims}
In this section the CVI and GSA-R algorithms are explored in simulation. The first example compares the performance of VI and CVI, and the second demonstrates GSA-R for clustering. The final scenario is an application where agents self-assign to channels subject to bandwidth and costs constraints.

\subsection{VI vs CVI Examples}
Figures \ref{fig-nonsep} and \ref{fig-sep} compare CVI and VI on non-separable and separable systems, respectively. Both systems had $N=7$ agents with binary choices, giving a state space a size of $2^7$. The CP had $3$ actions that were assigned per-cluster, yielding an action space of size $3^C$. Transition matrices and reward functions of the form $r(x)$ were randomly generated.

The graphs display $\|V^*\|$ averaged over all clustering assignments for each fixed $C$ and normalized by $\|V^*_{C=7}\|$. The value for $C=1$ is a baseline to demonstrate the increase in value attained with a clustered policy. The black dot on shows $\|\Vhat\|$ achieved via CVI, again averaged across all clustering assignments, normalized, and shown with error bounds. Clearly, CVI finds optimal values for separable systems. 

\begin{figure}
  \ifarxiv
  \centerline{\includegraphics[width=0.5\linewidth]{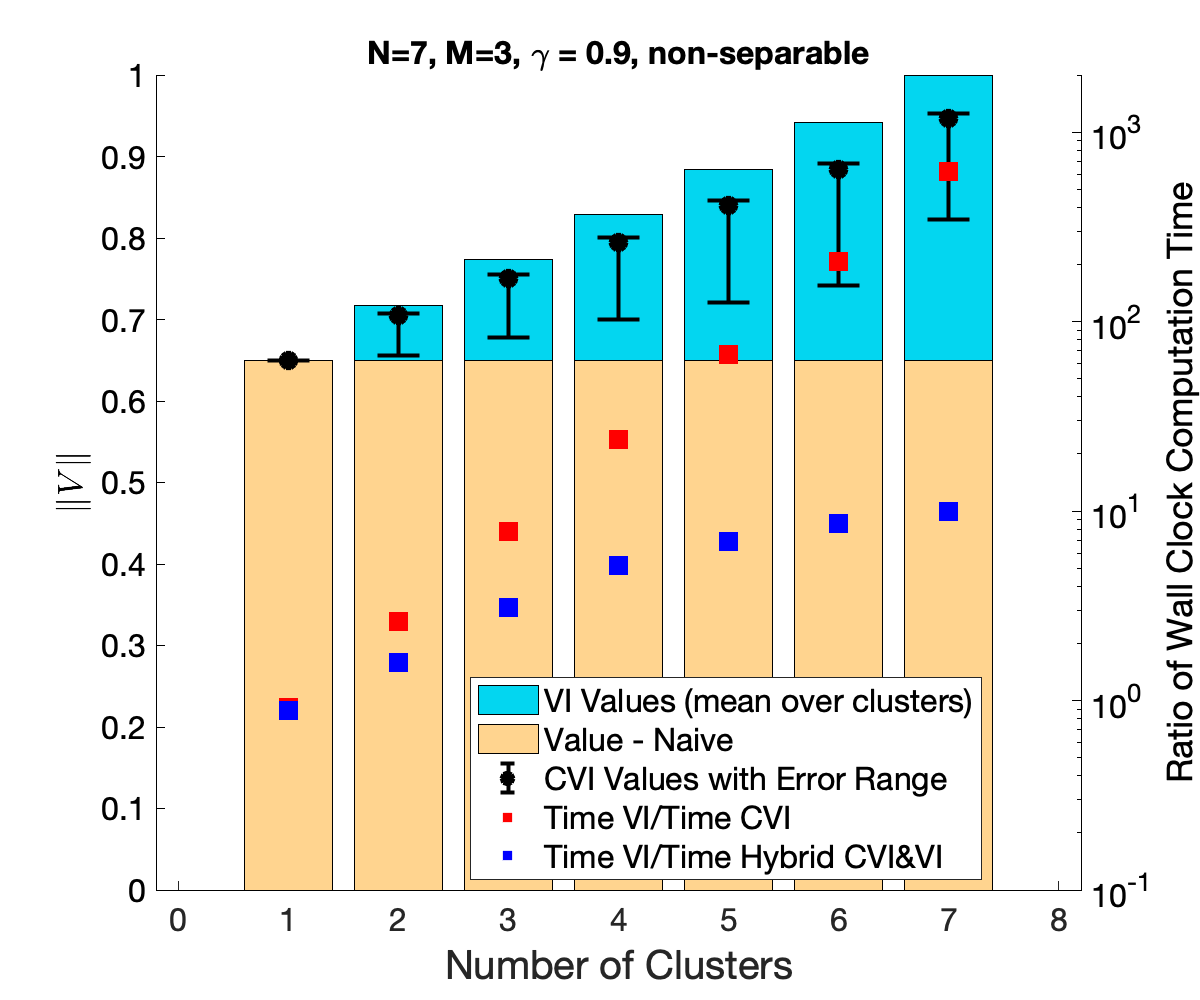}}
  \fi
  \caption{Performance of cluster-based control policies on a non-separable system for different $C$. Note that cluster-based control improves the attained value, and CVI is able to approximate the true optimal value with significant computation time savings. }
  \label{fig-nonsep}
\end{figure} 

\begin{figure}
  \ifarxiv
  \centerline{\includegraphics[width=0.5\linewidth]{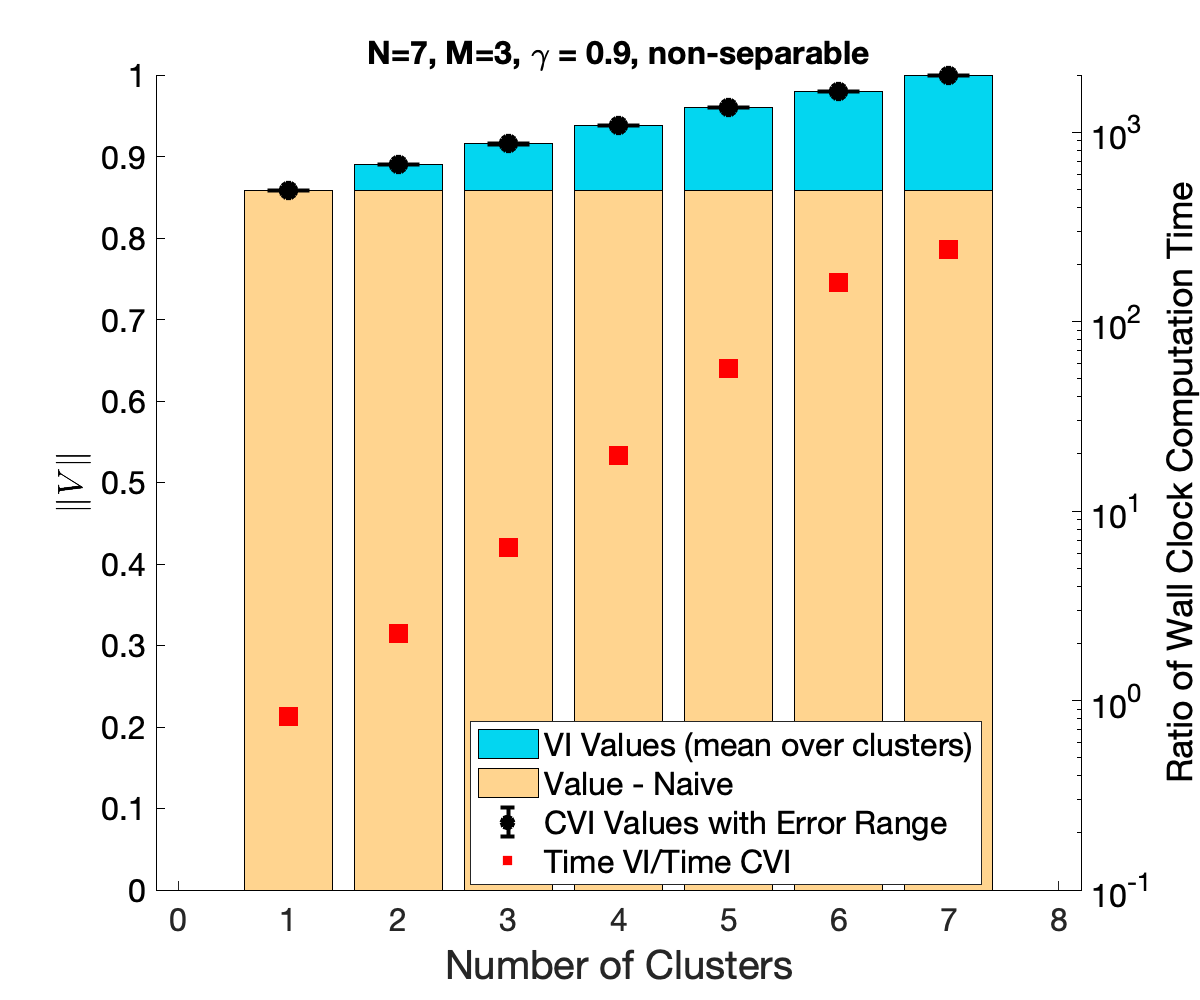}}
  \fi
  \caption{Performance of cluster-based control policies on a separable system for different $C$. Note that in for TI-MDPs with separable reward functions, CVI recovers the true optimal value exponentially faster than standard VI.}
  \label{fig-sep}
\end{figure}

The simulations were written in MATLAB, with the only difference in implementation being that VI optimized over the whole action space at each iteration, whereas CVI optimized over a subset of actions. The red squares show the ratio of the wall clock computation time for VI and CVI. These results confirm that CVI's computation time is independent of the number of clusters; VI depends exponentially. For example, consider Fig \ref{fig-nonsep} for $C=7$, where CVI converged about 620x faster than VI. In total computation time, CVI for $C=7$ was 1.2x slower than CVI for $C=1$, whereas VI for $C=7$ was 825x slower than VI for $C=1$.

Finally, the blue squares show the ratio of the wall clock computation for VI and hybrid CVI/VI. The stopping conditions used were $\delta=10^{-4}$, $\epsilon=10^{-5}$. The average number of VI calls used in the hybrid implementation is shown in Fig \ref{fig:vi in hybrid}. The time of hybrid CVI/VI scaled polynomially with respect to the number of clusters, whereas CVI scaled quadratically. 

\begin{figure}
    \centering
    \begin{tabular}{ |c|c|c|c|c|c|c| } 
     \hline
      \small{$C=1$} & \small{$C=2$} & \small{$C=3$} & \small{$C=4$} & \small{$C=5$} & \small{$C=6$} & \small{$C=7$} \\\hline
      2 & 3.12 & 3.36 & 3.70 & 3.87 & 4.0 & 4 \\
     \hline
    \end{tabular}
    \caption{Mean number of full VI iterations used in the hybrid CVI/VI method, where the mean is taken over all the clustering assignments.}
    \label{fig:vi in hybrid}
\end{figure}

\normalsize
\subsection{Greedy Clustering}
This next example demonstrates clustering via greedy splitting. Fig \ref{fig-gc} shows the results from a randomly generated separable TI-MDP constructed with $N=10$ agents. The bars show normalized $\|V^*\|$ produced by clustering assignments selected by GSA-R. Note that this sequence of values displays the diminishing returns property, thus demonstrating the results of Theorem \ref{submodular}. 

The black dots here are the ratio of the number of clustering assignments checked by a naive exhaustive search versus GSA-R. The naive method has $S(N,C)$ total combinations for each $C$ number of clusters. For GSA-R, only refinements of the previous clustering assignment are considered; for $C= 3,\dots, 8$ this leads to a reduction in the number of evaluations. For $C>8$, solving directly via the naive method is more efficient as $S(N,C)$ is small. These results demonstrate that GSA-R is a structured approach for partitioning the agents.


\begin{figure}
\centering
  \ifarxiv
  \includegraphics[width=0.5\linewidth]{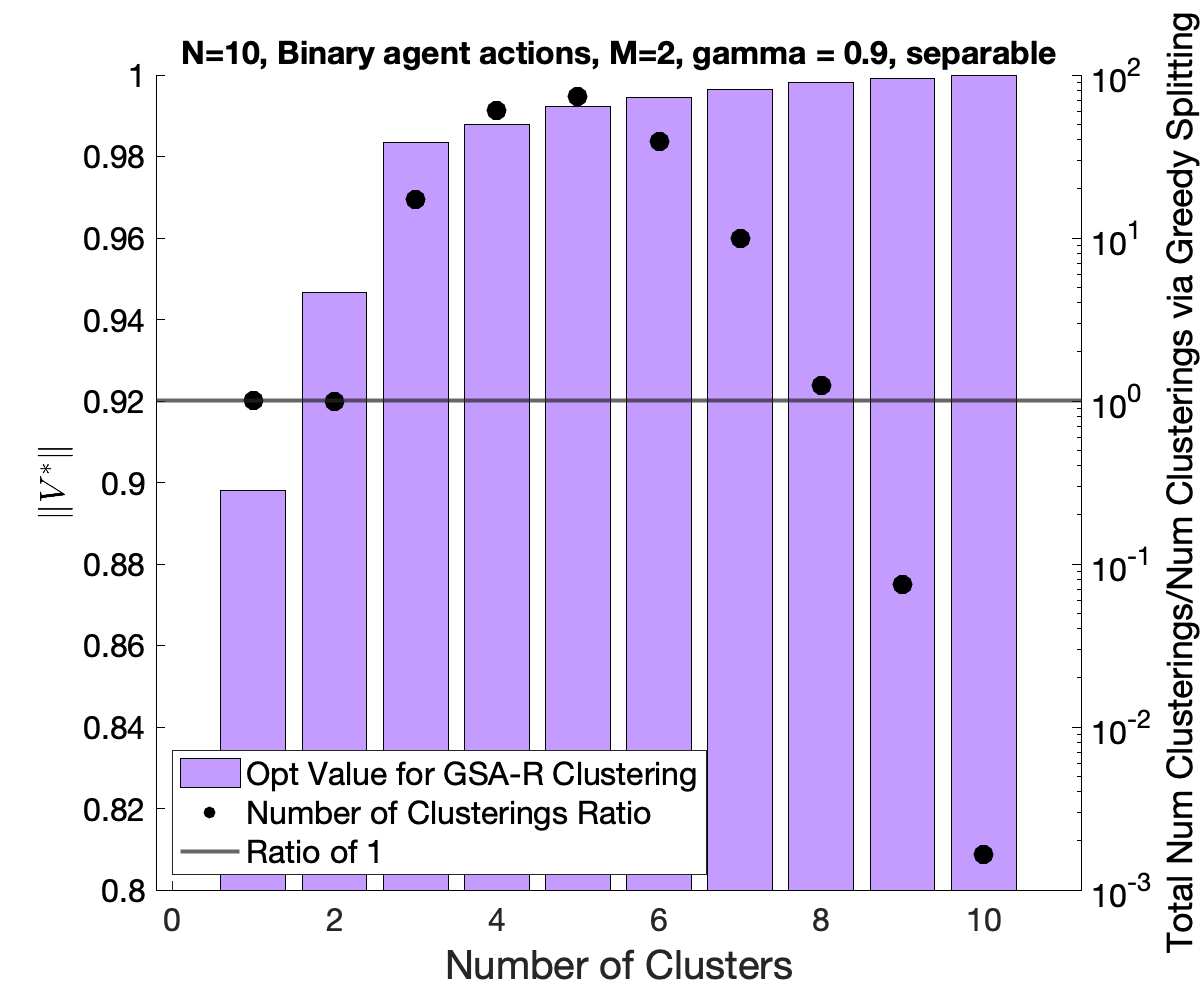}
  \fi
  \caption{Performance of greedy cluster splitting.}
  \label{fig-gc}
\end{figure} 

\subsection{Channel Assignment Example}
We consider an example where agents assign themselves to channels subject to bandwidth constraints and costs. This scenario will exemplify how a MAS, described as a game, can be abstracted to a MDP framework and controlled.

Let there be channels of low, medium, and high bandwidths (20, 50, 100). Each of $6$ agents selects a channel to use at each time step by evaluating their utility function: $u_n(x) = \frac{b(x)}{1+\sum_{n'\neq n} \mathbbm{1}(x_{n',t}=x)} - \beta_n\nu(x)$ where $b(x)$ is the bandwidth of channel $x$, $\nu(x)$ is the cost, and $\beta_n$ is a random scaling factor in $[0,1]$. This utility function can be interpreted as the effective bandwidth for $n$ had they chosen $x$ minus a scaled cost. All agents are assumed to use a best-response decision process, leading to a large sparse transition matrix. The cost structure, as chosen by the controller, can be one of:

\footnotesize
\begin{center}
\begin{tabular}{ |c|c|c|c| } 
 \hline
 $\nu(x, \alpha)$ & Low BW & Med BW & High BW \\\hline
 $\alpha_1$ & 0 & 10 & 30 \\\hline 
 $\alpha_2$ & 1 & 15 & 50 \\\hline
 $\alpha_3$ & 10 & 50 & 100\\
 \hline
\end{tabular}
\end{center}
\normalsize

The CP selects between (1) Maximum Revenue: The controller aims to maximize their payments from the agents, so $r(x,\abold) =\sum_{n\in\N}  \nu(x_n,\abold_n)$ which is agent-separable. (2) Desired Configuration: The controller wants all agents in the medium bandwidth channel. Their reward function is thus $r(x,
\abold) = \sum_{n\in\N} \mathbbm{1}(x_n = \text{medium})$, which again is separable.

The results of these scenarios are shown in Figure \ref{fig-bw}. Each row is normalized by the maximum value achieved when $C=6$. First, we notice that cluster-based policies improves the performance. The separable reward function enables CVI to find optimal control policies given a clustering assignment chosen by GSA-R. These examples demonstrate scenarios in which agent-separable reward functions enable desirable outcomes. Finally, note that GSA-R implies a natural stopping criterion. The CP can observe the negligible value improvement from two to three clusters, and then from submodularity conclude that the policy for three clusters is adequate.

\begin{figure}[ht]
    \centering
    \footnotesize\begin{tabular}{ |c|c|c|c|c|c|c| } 
     \hline
      $\frac{\|V^*\|}{\|V^*_{C=6}\|}$ & C=1 & C=2 & C=3 & C=4 & C=5 & C=6  \\\hline
      Scenario 1& 0.840  &  0.975   & 0.995  &  0.999 &   0.999   & 1  \\\hline
      Scenario 2 &0.882 &   0.984  &  0.993  &  0.999  &  0.999 &   1\\
     \hline
    \end{tabular}
    \caption{Results of the channel assignment example using CVI and GSA-R on two different separable reward functions.}
    \label{fig-bw}
\end{figure}

\section{Conclusion}
In this work we consider a CP who controls a MAS via cluster-based policies. Agents assigned to one cluster receive the same control at the same time step, but agents in different clusters may receive different controls. Complexity is a concern as standard techniques to solve for a policy scale with the size of the action space. We proposed clustered value iteration, which takes a ``round-robin" approach to policy optimization across each cluster in turn. This algorithm has a per-iteration complexity that is independent of the number of clusters, and has the same contractive properties of value iteration. We showed that our algorithm converges and it finds optimal policies for TI-MDPs with separable reward functions.

We next examined maximizing the value function by optimally assigning agents to clusters. We showed that TI-MDPs with submodular reward functions have submodular value functions. An iterative greedy splitting technique was then proposed that provided submodular improvement in value with natural stopping conditions and reduced computation versus a naive approach. Examples of both CVI and the GSA-R were explored in simulation under a variety of scenarios. 

While inspired by MAS, this work is relevant to general TI-MDPs with large action spaces. This work furthers understanding of TI-MDP structural properties by showing how the value function and Bellman operator may be split along clusters. Future work can study clustering agents by agent similarity metrics, and look at agents with active learning processes.

\bibliographystyle{IEEEtran}
\bibliography{IEEEabrv,main}

\end{document}